\documentclass[a4paper, 12pt]{amsart}

\usepackage[a4paper,left=1in,right=1in,top=1.25in,bottom=1.25in]{geometry}
\usepackage{amssymb}
\usepackage{graphicx}
\usepackage[mathlines]{lineno}

\newtheorem{theorem}{Theorem}[section]
\newtheorem{corollary}[theorem]{Corollary}
\newtheorem{lemma}[theorem]{Lemma}
\theoremstyle{definition}
\theoremstyle{remark}

\begin{document}
\title[]{What causes the increase in aggregation as a parasite moves up a food chain?}
\thanks{RM is supported by the Australian Research Council (Centre of Excellence for Mathematical and Statistical Frontiers, CE140100049). \\ 
We thank Dr Derek Zelmer, University of South Carolina, for early discussions. \\ Author contribution statement: RL initiated the study, collected the data and drafted the manuscript. RM  developed and analysed the model, and helped draft the manuscript. Both authors gave final approval for publication.}
\maketitle
\noindent R.J.G. LESTER and R. McVINISH\\ 
School of Biological Sciences, University of Queensland \\
School of Mathematics and Physics, University of Queensland \\

\noindent ABSTRACT. General laws in ecological parasitology are scarce.  Here we evaluate data published by over 100 authors to determine whether the number of hosts in a life cycle is associated with the degree of aggregation of fish parasites at different stages.  Parasite species were grouped taxonomically to produce 20 or more data points per group as far as possible.  Most parasites that remained at one trophic level were less aggregated than those that had passed up a food chain. We use a stochastic model to show that high parasite overdispersion in predators can be solely the result of the accumulation of parasites in their prey.  The model is further developed to show that a change in the predators feeding behaviour with age may further increase parasite aggregation. \\

\noindent Keywords: fish, index of dispersion, Monogenea, Nematoda, Taylor's power law, Acanthocephala.

\section{Introduction}

One of the few generalisations in ecological parasitology is that the frequency distributions of parasites are usually overdispersed, that is, a parasite population tends to be aggregated within certain host individuals \cite{Crofton:71}.  Most free-living organisms are aggregated in the environment but parasites are an extreme case, almost always highly aggregated in their host populations.  Understanding the processes that produce this heterogeneity in the distribution of macroparasites in their host populations continues to be a central research area in ecological parasitology \cite{WBDMPRRS:02}.  Theoretical studies have shown that demographic stochasticity produces aggregated distributions of species \cite{AGCH:82,Galvani:03}.  Experimental studies demonstrate that a range of factors influence the level of aggregation in infected hosts including spatial aggregation in infective stages \cite{KA:79},  host behaviour \cite{JH:14}, and host body condition and food availability \cite{TDS:13,LGH:14}.  Empirical studies have linked other factors, such as burrow structure in rabbit fleas \cite{KSKMSHM:06}, and season in fish strigeids \cite{JK:88}.  Shaw et al. \cite{SGD:98} summarised the proposed biological explanations as (i) a series of random infections with different densities of infectious stages, (ii) host individuals vary in susceptibility to infection, and (iii) non-random distribution of infective stages in the habitat.  Though Poulin \cite{Poulin:07} considered the main factors were heterogeneity in exposure and heterogeneity in susceptibility, Dobson and Merenlender \cite{DM:91} found that there was a tendency for lower levels of aggregation to be observed in intermediate hosts compared to definitive hosts, and Shaw and Dobson \cite{SD:95} in an analysis of several hundred host/parasite systems noted that those infections where definitive hosts were infected by consuming invertebrate intermediate hosts were associated with relatively high degrees of aggregation.   Lester \cite{Lester:12} provided evidence from his data on marine fish parasites that moving up from one trophic level to another made a significant contribution to the level of overdispersion in the parasite species.  

Moving up a food chain is an essential component of the life cycle of many aquatic parasites.  Here we examine aggregation levels for aquatic parasites reported  in the literature to see if indeed there is a general association between aggregation and trophic level.  We then consider the theoretical basis for the result.

\section{Methods}

\subsection{Empirical data}

The main measure for the degree of aggregation in samples of aquatic parasites reported in the literature is the variance (or standard deviation or standard error).  From this an `index of dispersion' can generally be calculated which is the variance divided by the mean \cite{CL:66}.  In parasite populations the index varies with the mean \cite{Scott:87,Poulin:93} so to provide a comparative figure, the index was calculated for a mean of 10 parasites per host using Taylor's power law \cite{Taylor:61}.

Values were extracted from over 300 papers, from the Journal of Parasitology 1998 -- July 2015, International Journal for Parasitology 1995 -- July 2015, Journal of Fish Biology, 1998 -- July 2015, Journal of Helminthology 2000 -- July 2015 and Comparative Parasitology 2000 -- July 2015, plus \cite{BM:55,AKKK:97,EEABSM:15}, and excluding data used in \cite{Lester:12}. The data is available in the electronic supplementary material.  Values were from wild fish in sample sizes of 20 or greater.  Data from mean abundances below 1 were excluded except for the two groups in invertebrates where no data was found with means of 1 or greater \cite{Crofton:71,Weinstein:72,Pearre:76,JK:85}.  Bird and mammal data were from \cite{RLC:99,SSOM:03,MAA:11,VMGRFL:11,RSKKP:13,SPCC:15}. 

The 1000+ variances were subdivided into taxonomic groups so that there were at least 20 data points in each group as far as possible.  The log variances were plotted against the log means (henceforth referred to as the log variance/log mean graph) for each group to reveal their linear relationship \cite{Taylor:61}. The log variance at a log mean abundance of 1 (that is, a mean of 10 parasites per fish) was estimated from the relationship and the index of dispersion derived.

Where intensity data was provided rather than abundance (that is, where zeros were excluded), means and variances were considered as abundances when prevalence was 95\% or greater.  Where data was at variance with other data in the publication, an email was sent to the author establish the correct information.  In addition, two samples from \cite{RHH:95} were excluded.  In one, one fish had 100 parasites and the other 64 fish had zero. In a second, one fish had 580 parasites and the other 30 fish had a mean of 1.   These extremes were far outside the rest of the data and were considered not to represent any trend.

\subsection{Stochastic model of parasite acquisition}

We consider a simple ecosystem consisting of parasites, prey fish and predator  fish.  To facilitate the development of the model, the ecosystem is assumed to be in equilibrium so the population sizes and age structures remain constant.  This may not be applicable in all the empirical data because samples from different areas have sometimes been combined and there may be long term changes in the host populations \cite{TBDWCP:11}.  A second assumption is that fish do not acquire immunity to the parasites.  Many of the parasite groups referred to accumulate in fish as the fish age, for example  \emph{Anisakis simplex} in salmon \cite{BM:55}, \emph{Tentacularia coryphenae} in skipjack tuna \cite{LBH:85}, \emph{Grillotia branchi} and \emph{Otobothrium cysticum} in \emph{Scomberomorus commerson} \cite{WL:06} suggesting little immunity against the parasites in these hosts.  A third assumption is that there is no parasite-associated host mortality.  Evidence for such mortality in wild fish has been difficult to obtain.  In most groups there appears to be little effect, and apart from some notable exceptions \cite{FKNRJK:11}, the concept that parasites tend to evolve to minimise host mortality is widely held \cite{LM:94}.  A final assumption is that there is no parasite mortality.  A few of the parasites considered are thought to have life spans less than that of the host, such as adult acanthocephalans \cite{Moller:76}, gyrodactylids \cite{SN:84} and some adult digenea \cite{MB:69}.  Many of the others such as larval trypanorhynchs and juvenile anisakids are thought to survive in fish for years \cite{LBH:85} suggesting parasite mortality in these groups may be minimal.  The effects of host and parasite mortality on parasite distributions has been investigated by others \cite{AM:78,Isham:95,BP:00} and are not incorporated into the model here.  

The model below first assumes that a predator encounters a random member of the prey population at times following a non-homogeneous Poisson process. In many predator species, the prey size changes as the predator ages \cite{ST:05,GKFM:09,SDPF:14,LGL:15} and hence the availability of a particular parasite distributions also changes. The model, therefore, is further developed to clarify the effect of a change in prey age on the parasite distribution in the predator.

Let $ X_{t} $ denote the parasite burden of a prey aged $ t $. The prey are born free of parasites so $ X_{0} = 0 $.  As parasites are assumed to survive for the lifetime of the host, $ X_{t} $ is a non-decreasing integer valued stochastic process. The first and simplest example is of a prey that takes infective particles at random, that is following a Poisson process. In this case, the parasite burden of a prey aged $ t $ has a Poisson distribution and, for any $ h > 0 $,
\[
P\left(X_{t+h} - X_{t} = k \mid X_{s}, 0\leq s\leq t\right) = \exp\left[-(\Lambda(t+h) - \Lambda(t)) \right] \frac{\left(\Lambda(t+h) - \Lambda(t)\right)^{k}}{k!}, \label{PS:Eq0}
\]
where $ \mathbb{E}(X_{t}) = \Lambda(t) $. 

Let $ Y_{t} $ denote the parasite burden of a predator aged $ t $. The predators are also born free of parasites so $ Y_{0} = 0 $. To allow for changes due to season and life cycle, the predator is assumed to encounter a random member of the prey population at times following a non-homogeneous Poisson process with intensity function $ \phi: [0,\infty) \mapsto [0,\infty) $. Not all encountered prey are consumed, rather encountered prey are consumed with a probability that depends on the age of the prey and possibly the age of the predator. 

Suppose that prey age in the population has a probability density function, which is denoted by $ f_{A}$. When a predator aged $ t $ encounters a prey aged $ u $, the probability that the predator consumes the prey is given by the function $ p(u,t) $. The times at which the predator consumes a prey therefore follows a Poisson process with intensity function $ \psi(t) =  \phi(t)\int^{\infty}_{0} p(u,t) f_{A}(u)du $.

The random variable $ \tilde{A}_{t} $ represents the age of a prey conditional on being consumed by a predator aged $ t $. The probability density function of $ \tilde{A}_{t} $ is
\begin{equation}
f_{\tilde{A}_{t}} (a) = \frac{ p(a,t) f_{A}(a)}{\int_{0}^{\infty} p(u,t) f_{A}(u)du}. \label{Dist:Eq1}
\end{equation}
If the prey is consumed, then the predator is assumed to acquire all parasites that are present in the prey. Let $ \tilde{X}_{t} $ denote the parasite burden of prey consumed by a predator aged $ t $,  that is $ \tilde{X}_{t} = X_{\tilde{A}_{t}} $. Using standard conditioning arguments, the mean and variance of the predator's parasite burden are given by
\begin{eqnarray}
\mathbb{E}(Y_{t})  & = & \int^{t}_{0} \mathbb{E}(\tilde{X}_{s}) \psi(s) ds \label{PS:Eq1} \\
\mbox{var}(Y_{t}) & = & \int^{t}_{0} \left(\mbox{var}(\tilde{X}_{s}) + \mathbb{E}(\tilde{X}_{s})^{2}\right) \psi(s)  ds \label{PS:Eq2b} \\
& = & \int^{t}_{0} \left(\frac{\mathbb{E}(\tilde{X}^{2}_{s})}{\mathbb{E}(\tilde{X}_{s})}\right)  \mathbb{E}(\tilde{X}_{s})  \psi(s) ds. \label{PS:Eq2}
\end{eqnarray}
When the distribution of $ \tilde{X}_{t} $ does not depend on $ t $, equation (\ref{PS:Eq2b}) reduces to the law of total variance.

To analyse this model, stochastic ordering properties and in particular the likelihood ratio ordering of random variables are used. Let $ U $ and $ V $ both be either continuous or discrete random variables with $ p_{U} $ and $ p_{V} $ denoting their respective probability densities, if they are continuous, or probability mass functions, if they are discrete. The random variable $ U $ is said to be smaller than $ V $ in the likelihood ratio ordering, denoted $ U \leq_{lr} V $, if $ p_{V}(w)/p_{U}(w) $ is an increasing function of $ w $ over the union of the supports of $ U $ and $ V $ \cite[Section 1.C.1]{SS:07}. To make the model more suited to exploiting the properties of stochastic ordering, we impose the following assumptions.
\begin{itemize}
\item[(A)] For any $ s \leq t $, $ X_{s} \leq_{lr} X_{t} $. 
\item[(B)] For any $ s < t $, the ratio 
\begin{equation} \label{StochOrder1}
\frac{p(u,t)}{p(u,s)} 
\end{equation}
is a non-decreasing function of  $u $. 
\end{itemize}
Assumption (A) is a technical assumption. It is satisfied by the Poisson process model for a prey's parasite burden as well as a number of other non-decreasing integer valued stochastic process such as certain mixed Poisson processes \cite{Grandell:97} and the negative binomial L\'{e}vy process \cite{KP:09}. Assumption (B) may be interpreted as follows: Suppose a prey were consumed by one of two predators and that both predators were equally likely to encounter the prey. Then the older predator becomes more likely to have consumed the prey as the age of the prey increases. It follows immediately from Assumption (B) that for any $ s \leq   t $, $ \tilde{A}_{s} \leq_{lr} \tilde{A}_{t} $. Therefore, older predators tend to consume older prey than younger predators in the sense that, if $ s \leq t $, then $ \Pr ( \tilde{A}_{t} > a) \geq \Pr (\tilde{A}_{s} > a ), $ for all $ a \geq 0 $ \cite[Theorem 1.C.1]{SS:07}. 

\section{Results} 

\subsection{Empirical data}

Data was taken from publications by over 200 authors. The results were relatively consistent, especially  for individual species or closely related species groups. For example, the relationship between the log means and log variances of cystacanths of \emph{Corynosoma} spp. from 9 senior authors and 43 fish samples showed a good correlation  ($R^{2}=0.91$, Fig. \ref{Fig0}).  The log variance at a log mean of 1.0 gave an index of dispersion of 20.4. 

Indices of dispersion for 28 parasite groups are given in Table 1.  Seven of the eight groups with the lowest scores are parasites that have not moved up a food chain, the exception being the hemiurid from chaetognaths which may or may not have been in an earlier host. Of the remaining 20 parasite groups, all but, the diplectanids, ancyrocephalids and metacercariae,  are thought to have arrived in the host through being eaten in a prey item. The dispersion indices are consistent with the premise that moving up a food chain is frequently associated with an increase in overdispersion.

Not included in the table is a high index for \emph{Gyrodactylus} spp., 43.7 (from 16 data points from 3 papers).  Gyrodactylids present a special case as the parasites reproduce on the fish. Also omitted are the catchall groups `fish metacestodes excluding \emph{Grillotia} and diphyllobothriids' and `other larval Acanthocephala'  The data from these were highly variable ($R^{2}= 0.38 $ and $ 0.60$).

\subsection{Model: Fixed prey selection} \label{Sec:FPS}
Suppose that the predator's selection of prey remains constant throughout its life so the ratio (\ref{StochOrder1}) does not depend on $ u $. Since the distribution of $ \tilde{X}_{t} $ no longer depends on $ t $, we let $ \tilde{X} $ denote the parasite burden of the consumed prey. Let $ \Psi(t) = \int^{t}_{0} \psi(s) ds $, which is the expected number of prey consumed during the period $ [0,t] $. From equations (\ref{PS:Eq1}) and (\ref{PS:Eq2}), $ \mathbb{E}(Y_{t}) = \mathbb{E}(\tilde{X}) \Psi(t) $ and 
\begin{equation}
\mbox{var}(Y_{t}) = \mathbb{E}(\tilde{X}^{2}) \Psi(t) = \frac{\mathbb{E}(\tilde{X}^{2})}{\mathbb{E}(\tilde{X})}  \mathbb{E}(Y_{t}).  \label{FPS:Eq0}
\end{equation}
Equation (\ref{FPS:Eq0}) implies the index of dispersion for the predator's parasite burden is constant throughout its life when prey selection is fixed. Hence, the feeding rate of the predator has no effect on the aggregation in its parasite burden. The log variance/log mean graph is a straight line with a slope of one. The intercept is given by $ \log (\mathbb{E}( \tilde{X}^{2})/\mathbb{E}(\tilde{X})) $, and is typically positive. For example, if $ X_{t} $ is a mixed Poisson process or if $ \mathbb{E}(\tilde{X}) > 1 $, then the intercept is positive.

To see the effect of trophic level in this setting, suppose there is a second predator in the ecosystem that consumes the original predator. Assuming the second predator selects its prey in the same manner throughout its life, equation (\ref{FPS:Eq0}) still applies with the original predator now acting as the prey. The log variance/log mean graph of the second predator's parasite burden is again a straight line with a slope of one. Let $ \tilde{Y} $ denote the parasite burden of the first predator given it has been consumed by the second predator, then
\begin{equation}
\frac{\mathbb{E}(\tilde{Y}^{2})}{\mathbb{E}(\tilde{Y})}  \geq \mathbb{E}(\tilde{Y}) + \frac{\mathbb{E}(\tilde{X}^{2})}{\mathbb{E}(\tilde{X})}. \label{FPS:Eq1} 
\end{equation}
The proof of inequality (\ref{FPS:Eq1}) is given in the appendix. Inequality (\ref{FPS:Eq1}) shows the intercept for the second predator is greater than the intercept of the first. If the prey accumulates parasites following a Poisson process, and the predators do not change their prey selection, then the parasite burden becomes more overdispersed as the parasite passes up a food chain. This is illustrated in Figure \ref{FigFPS}.

\subsection{Prey selection depending on predator age}
As noted previously, prey selection typically changes with the age of the predator. In this case, the relationship between the aggregation of parasites in predator and prey is more complex. The following result provides some basic inequalities relating the index of dispersion for the predator and prey.
\begin{theorem} \label{Thm2}
Suppose Assumptions (A) and (B) hold. Then 
\begin{equation}
\frac{\mbox{var}(Y_{t})}{\mathbb{E}(Y_{t})}  \leq \frac{\mathbb{E}(\tilde{X}^{2}_{t})}{\mathbb{E}(\tilde{X}_{t})}, \label{Thm2:Eq0}
\end{equation}
and
\begin{eqnarray}
\frac{\mbox{var}(Y_{t})}{\mathbb{E}(Y_{t})} & \geq & \frac{1}{\Psi(t)} \int^{t}_{0} \left(\frac{\mathbb{E}(\tilde{X}^{2}_{s})}{\mathbb{E}(\tilde{X}_{s})}\right) \psi(s) ds  \label{Thm2:Eq1}\\
& \geq &  \int^{\infty}_{0} \left(\frac{\mbox{var}(X_{r})}{\mathbb{E}(X_{r})} + \mathbb{E}(X_{r}) \right) \left(  \frac{1}{\Psi(t)} \int^{t}_{0} f_{\tilde{A}_{s}} (r)  \psi(s) ds \right) dr. \label{Thm2:Eq2}
\end{eqnarray}
\end{theorem}

Inequality  (\ref{Thm2:Eq0}) bounds the index of dispersion of the predator aged $ t $ in terms of the parasite burden of prey consumed at aged $ t $. This bound implies that an increase in the index of dispersion with trophic level may not be observed in the empirical data if the data is collected from prey that is older than what the predator typically consumes. The quantity $ \Psi(t)^{-1} \int^{t}_{0} f_{\tilde{A}_{s}} (r)  \psi(s)  ds $ appearing in inequality (\ref{Thm2:Eq2}) is the probability density of the age of a random selected prey that was consumed by a predator during the period $ [0, t] $. Inequality  (\ref{Thm2:Eq2}) shows that the index of dispersion of the predator is bounded from below by the average of the index of dispersion of the prey it consumed plus the average parasite burden of its prey. Equality holds in (\ref{Thm2:Eq0}) and (\ref{Thm2:Eq1}) if, for all $ s $ and $ t $, the ratio (\ref{StochOrder1}) does not depend on $ u $, that is, if the distribution of consumed prey's age does not depend on the age of the predator. 

In empirical data, the exact relationship between mean and variance is difficult to identify. In the previous section, we saw that our modelling assumptions imply that the slope in the log variance/log mean graph is one when prey selection is fixed. The following results show that allowing prey selection to depend on the predator's age results in the slope being greater than one.

\begin{theorem} \label{Thm1}
Suppose Assumptions (A) and (B) hold. Then 
\begin{equation}
\frac{ d \log \mbox{var}(Y_{t})}{d \log \mathbb{E}(Y_{t}) } \geq 1, \label{Thm1:Eq1}
\end{equation}
with equality if and only if for all $ s \in [0,t] $, $ \mathbb{P}(\tilde{X}_{t} = m \mid \tilde{X}_{t} > 0) = \mathbb{P}(\tilde{X}_{s} = m \mid \tilde{X}_{s} > 0 ) $ for all $ m \geq 1 $. 
\end{theorem}

\begin{corollary} \label{Cor1}
Suppose Assumptions (B) holds and $ X_{t} $ is a Poisson process. Then equality holds in (\ref{Thm1:Eq1}) if and only if the ratio (\ref{StochOrder1}) is constant in $ u $ for all $ s \in [0,t] $. 
\end{corollary}

Theorem \ref{Thm1} shows that the slope of the log variance/log mean graph will be greater than one unless the distribution of the parasite intensity in the prey remains constant throughout the predator's life. When the prey accumulates parasites following a Poisson process, this is only possible if the prey selection is fixed. Even if changes in prey selection are restricted to a small part of the predator's life, the slope of the log variance/log mean graph will be greater than one. 

The following example illustrates the effect of prey selection on the log variance/log mean graph. Suppose that parasites accumulate in the prey according to a mixed Poisson process \cite{Grandell:97} satisfying Assumption (A) with rate $ \lambda $ so $ \mathbb{E}(X_{t}) = \mathbb{E}(\lambda) t $ and $ \mbox{var}(X_{t}) = \mathbb{E}(\lambda) t + \mbox{var}(\lambda) t^{2}$. The age distribution of the prey is taken to be Gamma with shape and rate parameters $(\alpha,\beta) $. If the function $ p(a,t) $ is proportional to $ a^{\delta t} \exp(-a \gamma) $, $ \gamma, \delta > 0 $, then the ratio $ p(a,t)/p(a,s) $ is proportional to $ a^{\delta(t-s)} $, which is non-decreasing in $ a $ for $ s < t $. The distribution of $ \tilde{A}_{t} $ is then Gamma with shape and rate parameters $ (\alpha + \delta t, \beta + \gamma) $. Standard calculations show that
\begin{eqnarray}
\mathbb{E}(\tilde{X}_{t}) & = & \mathbb{E}(\lambda) \frac{\alpha + \delta t}{\beta + \gamma} \label{Eg:Eq1} \\
\frac{\mathbb{E}(\tilde{X}_{t}^{2})}{\mathbb{E}(\tilde{X}_{t})} & = & 1+ \frac{\mathbb{E}(\lambda^{2})}{\mathbb{E}(\lambda)} \left( 1+ \frac{\alpha + \delta t}{\beta + \gamma}  \right) \label{Eg:Eq2} 
\end{eqnarray}
In Figure \ref{Fig1}, the log variance/log mean graph is given 
for a range of values of $ \alpha $ and $ \delta $ with $ \mathbb{E}(\lambda) = 1,\ \mathbb{E}(\lambda^{2}) = 1.5,\ \beta+\gamma=1 $ and $ \psi(t) = 1 $ for all $ t \geq 0 $. 
As the mean increases, the slope of the graph appears to be independent of the parameter values chosen. For large $ t $, 
\begin{eqnarray*}
\mathbb{E}(\tilde{X}_{t}) \sim \frac{\mathbb{E}(\lambda) \delta t}{\beta + \gamma}, \quad  \mathbb{E}(\tilde{X}_{t}^{2}) \sim  \mathbb{E}(\lambda^{2}) \left(\frac{\delta t}{\beta + \gamma}\right)^{2}. 
\end{eqnarray*}
With $ \psi(t) = 1 $,  the slope of the log variance/log mean graph approaches $ 3/2 $  when the mean is large.

\subsection{Role of target prey size}
In Section \ref{Sec:FPS} it was observed that if the predator selects its prey in the same manner, regardless of age, then the feeding rate of the predator has no effect on the aggregation in its parasite burden. However, when the predator's prey selection depends on the age of the predator, the feeding rate of the predator can have a considerable impact on the aggregation of its parasite burden. 

To examine the effect of the feeding rate on parasite aggregation, suppose there are two predator species. Quantities relating to the two predators are distinguished through subscripts. The two predators are assumed to consume on average the same amount of biomass, but their feeding patterns differ in that the first species tends to consume many smaller, hence younger, prey fish while the second species consumes few larger, hence older, prey fish. Though fish growth typically slows with age there is generally a strong relationship between size and age \cite{Sale:02}. Assuming the prey accumulates parasites at a constant rate, the expected biomass of prey aged $ t $ is proportional to its expected parasite burden. The following result shows that while predator one has a higher feeding rate than predator two, the parasites are more aggregated in predator two than in predator one.

\begin{theorem} \label{Thm3}
Suppose that Assumptions (A) and (B) hold and that $ \mathbb{E}(X_{t}) \propto t $ for all $ t \geq 0 $. Assume also that, for all $ t \geq 0 $, $ \tilde{A}_{1,t} \leq_{lr} \tilde{A}_{2,t} $  and 
\begin{equation}
\mathbb{E} (\tilde{A}_{1,t}) \psi_{1}(t) = \mathbb{E} (\tilde{A}_{2,t}) \psi_{2}(t) \label{FR:Eq1}
\end{equation}
Then $ \psi_{t}(t) \geq \psi_{2}(t) $ and 
\begin{eqnarray*}
\frac{\mbox{var}(Y_{1,t})}{\mathbb{E}(Y_{1,t})} \leq \frac{\mbox{var}(Y_{2,t})}{\mathbb{E}(Y_{2,t})} 
\end{eqnarray*}
for all $ t \geq 0$.
\end{theorem}

Suppose we take the same scenario as used in Figure \ref{Fig1} except that the function $ p(a,t) $ is now proportional to $ a^{\delta(t)} \exp(-a \gamma) $, where $ \delta(t) $ is an increasing function and $ \gamma > 0 $. Equations (\ref{Eg:Eq1}) and (\ref{Eg:Eq2}) hold with $ \delta t $ replaced by $ \delta(t) $. In Figure \ref{Fig2}, $ \mathbb{E}(\lambda) =1,\ \mathbb{E}(\lambda^{2})=1.5,\ \beta+\gamma = 1$ and $ \alpha = 0.1 $ for both predators. For the first predator we set $ \delta_{1}(t) = t $ and $ \psi_{1}(t) = 1 $ for all $ t \geq 0$, and for the second predator $ \delta_{1}(t) = t + t^{2} $ and $ \psi_{2}(t) = (0.1 + t)/(0.1 + t + t^{2}) $. As the mean parasite burden increases, the slope of the log variance/log mean graph for the second predator approaches $ 2 $ whereas for the first predator the slope approaches $ 3/2$.

\section{Discussion}

In the empirical data, all the fish parasite distributions were overdispersed, even those with a single host in the life cycle.  The latter group had indices around 10 suggesting that major sources of aggregation other than host number are operating on the parasites.  Nevertheless the ordination of parasite groups according to their index shows that those that pass through earlier hosts tend to be more overdispersed than monoxenous parasites, a result in agreement with Lester \cite{Lester:12} who used a different data set.

Insufficient data were available to permit the evaluation of the degree of overdispersion for one parasite species at different stages of its life cycle.  However, trends can be discerned.  The index  for cystacanths in gammarids was estimated to be 2.7 (Table \ref{Tab1}). The indices for \emph{Corynosoma} cystacanths in fish was 20.4 and for \emph{Corynosoma} adults in seals was 24.0, thus though data from seals is sparse, the results conform to a general increase in the index of dispersion at each stage of the life cycle.  Acanthocephalans as adults in fish also had a high index (23.4).

Hemiurids in planktonic invertebrates appeared to be almost randomly distributed, that is the variance was close to the mean, and had an index of 2.6.  Adult hemiurids in fish were more aggregated with an index of 22.9.  

Metacercariae in fish had high indices of 28.8 and 28.2.  That of adult  digeneans in fish eating birds and mammals was higher at 39.8.  The indices for the metacercariae could be a consequence of molluscs releasing clouds of cercariae so that contact with the cercariae by fish is not well modelled by a Poisson process.

The monogenean groups of capsalids, dactylogyrids and polyopisthocotyleans all had low indices (8.7, 10.0 and 11.0).  In contrast diplectanids and ancyrocephalids, also with one-host life cycles, had high indices (28.8 and 23.4).  Again the infection process in these groups may not be well modelled by a Poisson process.  Eggs of the diplectanid \emph{Allomurraytrema robustum}  become entangled among adults on the gills and the diplectanid \emph{Lamellodiscus acanthopagri} actually attaches eggs to the gills. Larvae that hatch from these eggs can attach to adjacent filaments thus greatly increasing the chances of the fish obtaining a subsquent infection.  In contrast the polyopisthocotylean \emph{Polylabroides multispinosus} sheds its eggs into the water column, and none of its eggs are attached to the host \cite{Roubal:95}.

The conclusions here concur with the findings of Dobson and Merenlender \cite{DM:91} and Shaw and Dobson \cite{SD:95} who suggested that there was a tendency for aggregation to be greater in definitive hosts compared to intermediate hosts.  Poulin \cite{Poulin:13} apparently using essentially the same data base as used here, failed to detect any differences in overdispersion between monoxenous and heteroxenous parasites, possibly because he incorporated data from samples with as few as 6 hosts per sample and used taxonomic groups that were more ecologically diverse than used here.

Taylor and Woiwod \cite{TW:82} and Anderson and Gordon \cite{AG:82} concluded that the values in the log variance/log mean relationship were characteristic of a species at a particular point in space and time.  For aquatic parasites the stage in the life cycle appears to be important.  For example, the combined slope for \emph{Anisakis simplex} from 6 authors for 6 different fish species was 1.74 (Table 1).  The slopes for \emph{Anisakis} 1 used in \cite{Lester:12}  were 1.47 for a crustacean-eating fish and 1.53, 1.81 and 1.97 for three piscivorous fishes.  The indices of dispersion were, 5.1, 20.2, 36.3 and 72.4 (from \cite{Lester:12}). The results are consistent with differences between host species for what appears to be the same parasite at different trophic levels as predicted by the model.  

The empirical data demonstrate a strong association between the level of aggregation of aquatic parasites and the trophic level of their hosts in relation to the parasite life cycles.   There is good support for such a link from mathematical theory. To confirm the association, data are required on the levels of aggregation of a single population of parasites as it moves up a food chain.

\section{Appendix: Proofs}

The proofs of equation (\ref{FPS:Eq1}) and Theorem \ref{Thm2} use a continuous version of the Chebyshev sum inequality. This inequality states that for any non-decreasing functions $ g $ and $ h $ on $ \mathbb{R} $ and any probability measure $ \mu $ on $ \mathbb{R} $
\begin{equation}
\int g(x) h(x) \mu(dx) \geq \int g(x) \mu(dx) \int h(x) \mu(dx). \label{CSI}
\end{equation}

\subsection{Proof of inequality (\ref{FPS:Eq1})}
Let $ f_{\tilde{A}} $ denote the probability density function of the age of prey consumed by the second predator. Then
\[
\frac{\mathbb{E}(\tilde{Y}^{2})}{\mathbb{E}(\tilde{Y})} = \frac{\int \mathbb{E}(Y^{2}_{s}) f_{\tilde{A}}(s) ds}{\int\mathbb{E}(Y_{s}) f_{\tilde{A}}(s) ds}.
\]
From equation (\ref{PS:Eq1}), $ \mathbb{E}(Y_{t}) $ is increasing, and since the ratio $ \mbox{var}(Y_{t})/\mathbb{E}(Y_{t}) $ is constant by equation (\ref{FPS:Eq0}), $ \mathbb{E}(Y^{2}_{t})/\mathbb{E}(Y_{t}) $ is also increasing. Applying inequality (\ref{CSI}) gives 
\begin{equation}
\frac{\mathbb{E}(\tilde{Y}^{2})}{\mathbb{E}(\tilde{Y})} \geq \int \frac{\mathbb{E}(Y^{2}_{s})}{\mathbb{E}(Y_{s})} f_{\tilde{A}}(s) ds = \int \left( \mathbb{E}(Y_{s}) + \frac{\mbox{var}(Y_{s})}{\mathbb{E}(Y_{s})} \right) f_{\tilde{A}}(s) ds. \label{FPS:Eq2}
\end{equation}
Inequality (\ref{FPS:Eq1}) follows by substituting the expression for $ \mbox{var}(Y_{t}) $ in equation (\ref{FPS:Eq0}) into inequality (\ref{FPS:Eq2}). 

\subsection{Proof of Theorem \ref{Thm2}}

As noted earlier, an immediate consequence of Assumption (B) is that $ \tilde{A}_{s} \leq_{lr} \tilde{A}_{t} $ for all $ s \leq t $.  From  Assumption (A) and \cite[Theorem 1.C.17]{SS:07}, $ \tilde{A}_{s} \leq_{lr} \tilde{A}_{t} $ implies that $ \tilde{X}_{s} \leq_{lr} \tilde{X}_{t} $. Therefore,  $ \mathbb{E}(\tilde{X}_{t}) $ is a non-decreasing function of $ t $, and by \cite[Theorem 1.C.20]{SS:07} so is $  \mathbb{E}(\tilde{X}^{2}_{t})/\mathbb{E}(\tilde{X}_{t}) $. Inequality (\ref{Thm2:Eq0}) now follows as
\begin{equation}
\frac{\mbox{var}(Y_{t})}{\mathbb{E}(Y_{t})} = \frac{\int^{t}_{0} \left(\frac{\mathbb{E}(\tilde{X}^{2}_{s})}{\mathbb{E}(\tilde{X}_{s})}\right)  \mathbb{E}(\tilde{X}_{s})  \psi(s) ds}{\int^{t}_{0} \mathbb{E}(\tilde{X}_{s})  \psi(s) ds}. \label{GPM:Eq0}
\end{equation}
Inequality (\ref{Thm2:Eq1}) is obtained by applying inequality (\ref{CSI}) to equation (\ref{GPM:Eq0}).
To prove inequality (\ref{Thm2:Eq2}), again apply inequality (\ref{CSI}) to the ratio 
\begin{equation}
\frac{\mathbb{E}(\tilde{X}^{2}_{s})}{\mathbb{E}(\tilde{X}_{s})}  =  \frac{\int^{\infty}_{0} \mathbb{E}(X_{s}^{2}) f_{\tilde{A}_{t}} (s) ds}{\int^{\infty}_{0} \mathbb{E}(X_{s}) f_{\tilde{A}_{t}} (s) ds}. \label{GPM:Eq1}
\end{equation}
This is possible since $ X_{s} \leq_{lr} X_{t} $ for all $ s \leq t $ by Assumption (A), and this implies that $ \mathbb{E}(X_{t}) $ and $  \mathbb{E}(X^{2}_{t})/\mathbb{E}(X_{t}) $ are both non-decreasing functions of $ t $  \cite[Theorem 1.C.20]{SS:07}. Therefore,
\begin{eqnarray}
\frac{\mathbb{E}(\tilde{X}^{2}_{s})}{\mathbb{E}(\tilde{X}_{s})} & \geq & \int^{\infty}_{0} \frac{\mathbb{E}(X_{s}^{2})}{\mathbb{E}(X_{s})} f_{\tilde{A}_{t}} (s) ds.  \label{GPM:Eq2}
\end{eqnarray}
Substituting the lower bound (\ref{GPM:Eq2}) into inequality (\ref{Thm2:Eq1}), we obtain
\begin{eqnarray*}
\frac{\mbox{var}(Y_{t})}{\mathbb{E}(Y_{t})} & \geq & \frac{1}{\Psi(t)} \int^{t}_{0} \left(\int^{\infty}_{0} \frac{\mathbb{E}(X_{r}^{2})}{\mathbb{E}(X_{r})} f_{\tilde{A}_{s}} (r) dr \right)  
 \psi(s) ds.
\end{eqnarray*}
Finally, interchanging the order of integration yields inequality (\ref{Thm2:Eq2}).

\subsection{Proof of Theorem \ref{Thm1}}

To prove Theorem \ref{Thm1}, we first give the following basic result.

\begin{lemma}\label{Lem2}
Assume that $ f: [0,\infty) \mapsto (0,\infty) $ and $ g: [0,\infty) \mapsto (0,\infty) $. Assume also that $ f $ is a non-decreasing function. Define
\begin{eqnarray*}
\mu_{t} & := & \int^{t}_{0} g(s) ds \\
\sigma^{2}_{t} & := & \int^{t}_{0} f(s) g(s) ds.
\end{eqnarray*}
Then $ \sigma^{2} $ may be expressed as a function of $ \mu $ and 
\begin{equation}
\frac{ d \log \sigma^{2}}{d \log \mu } = \frac{\mu}{\sigma^{2}} f \left( G^{-1}(\mu) \right), \label{Lem2:Eq1}
\end{equation}
where $ G^{-1} $ is the function such that $ G^{-1}(\mu_{t}) = t $. Furthermore, 
\begin{equation}
\frac{d \log \sigma^{2}}{d \log \mu} \geq 1,  \label{Lem2:Eq2}
\end{equation}
with equality if and only if for all $ s \in [0,G^{-1}(\mu)] $, $ f(s) = f(G^{-1}(\mu)) $.
\end{lemma}

\begin{proof}
As $ \mu_{t} $ is strictly increasing in $ t $, there exists a function $ G^{-1}:[0,\infty) \mapsto [0,\infty) $ such that $ G^{-1}(\mu_{t}) = t $. Hence, we may write $ \sigma^{2}(\mu) = \sigma^{2}_{G^{-1}(\mu)} $.  Equation (\ref{Lem2:Eq1}) follows from the chain rule. From the mean value theorem, there exists an $ s^{\ast} \in [0,G^{-1}(\mu)] $ such that $ \sigma^{2} / \mu = f(s^{\ast}) $. Inequality (\ref{Lem2:Eq2}) now follows as $ f $ is non-decreasing. For the equality in (\ref{Lem2:Eq2}) to hold, we must have $ f(G^{-1}(\mu)) \int^{G^{-1}(\mu)}_{0} g(s) ds = \int^{G^{-1}(\mu)}_{0} f(s) g(s) ds $. As $ f $ is assumed to be non-decreasing, this is only possible if $ f $ is constant on $ [0,G^{-1}(\mu)] $.
\end{proof}

Returning to the proof of Theorem \ref{Thm1}, we note that $  \mathbb{E}(\tilde{X}^{2}_{t})/\mathbb{E}(\tilde{X}_{t}) $ is a non-decreasing function of $ t $ (see the proof of Theorem \ref{Thm2}).  Let $ g(s) = \mathbb{E}(\tilde{X}_{s}) \psi(s) $ and $ f(s) = \mathbb{E}(\tilde{X}_{s}^{2}) / \mathbb{E}(\tilde{X}_{s}) $. Inequality (\ref{Thm1:Eq1}) now follows from Lemma \ref{Lem2}. Also from Lemma \ref{Lem2}, we see that equality in (\ref{Thm1:Eq1}) holds if and only if $ \mathbb{E}(\tilde{X}_{s}^{2}) / \mathbb{E}(\tilde{X}_{s}) = \mathbb{E}(\tilde{X}_{t}^{2}) / \mathbb{E}(\tilde{X}_{t})  $ for all $ s \in [0,t] $. This is only possible if $ \tilde{X}^{2}_{s} \tilde{X}_{t} \stackrel{d}{=} \tilde{X}_{s}\tilde{X}^{2}_{t} $ for all $ s \in [0,t] $  \cite[Theorem 1.C.20]{SS:07}. Let $ m $ be a square-free positive integer. Then 
\[
\mathbb{P}(\tilde{X}^{2}_{s} \tilde{X}_{t} = m ) = \mathbb{P} (\tilde{X}_{s} = 1) \mathbb{P}(\tilde{X}_{t} = m ).
\]
As $ \tilde{X}^{2}_{s} \tilde{X}_{t} \stackrel{d}{=} \tilde{X}_{s}\tilde{X}^{2}_{t} $, 
\begin{equation}
\frac{\mathbb{P}(\tilde{X}_{t} = m )}{\mathbb{P}(\tilde{X}_{s} = m )} = \frac{\mathbb{P} (\tilde{X}_{t} = 1)}{\mathbb{P} (\tilde{X}_{s} = 1)}, \label{Thm1:Proof:Eq1}
\end{equation}
for all square-free integers $ m \geq 1$ in the union of the supports of $ \tilde{X}_{s} $ and $ \tilde{X}_{t} $.  Since $ \tilde{X}_{s} \leq_{lr} \tilde{X}_{t} $, equation (\ref{Thm1:Proof:Eq1}) must hold for all $ m \geq 1$ in the union of the supports of $ \tilde{X}_{s} $ and $ \tilde{X}_{t} $. From equation (\ref{Thm1:Proof:Eq1}), 
\[ 
\mathbb{P}(\tilde{X}_{s} > 0) = \frac{\mathbb{P} (\tilde{X}_{t} = 1)}{\mathbb{P} (\tilde{X}_{s} = 1)} \mathbb{P}(\tilde{X}_{t} > 0).
\]
Therefore,  $ \mathbb{P}(\tilde{X}_{t} = m \mid \tilde{X}_{t} > 0) = \mathbb{P}(\tilde{X}_{s} = m \mid \tilde{X}_{s} > 0 ) $ for all $ m \geq 1 $.

\subsection{Proof of Corollary \ref{Cor1}}
As $ X_{t} $ is a Poisson process, $ \tilde{X}_{t} $ has a mixed Poisson distribution. Let $ P_{t}(z) $ be the probability generating function of $ \tilde{X}_{t} $. From equation (\ref{Thm1:Proof:Eq1}) there exists a $ c_{s,t} \in [0,1] $ such that 
\begin{equation}
P_{s}(z) = c_{s,t} P_{t}(z) + (1-c_{s,t}), \label{Cor:Eq1}
\end{equation}
for all $ z $. It is known that the mixing distribution of a mixed Poisson distribution is uniquely identifiable \cite{Feller:43}. Therefore, equation (\ref{Cor:Eq1}) implies that $ \tilde{A}_{s} \stackrel{d}{=} Z \tilde{A}_{t} $, where $ Z $ is an independent Bernoulli random variable with $ \mathbb{P}(Z=1) = c_{s,t} $. From equation (\ref{Dist:Eq1}), this is only possible if $ c_{s,t} = 1 $. Hence, the ratio (\ref{StochOrder1}) is constant.

\subsection{Proof of Theorem \ref{Thm3}}

As $ \tilde{A}_{1,t} \leq_{lr} \tilde{A}_{2,t} $, $ \mathbb{E} (\tilde{A}_{1,t}) \leq \mathbb{E} (\tilde{A}_{2,t}) $. Hence, equation (\ref{FR:Eq1}) implies  $ \psi_{1}(t) \geq \psi_{2}(t) $. As $ \mathbb{E}(\tilde{X}_{i,t}) = \mathbb{E}(\tilde{A}_{i,t}) $ for all $ t\geq 0 $ and $ i=1,2 $, equations (\ref{PS:Eq1}) and (\ref{FR:Eq1}) imply $ \mathbb{E}(Y_{1,t}) = \mathbb{E}(Y_{2,t}) $ for all $ t \geq 0$, that is, the expected parasite burden of the two predators are equal for all ages. As $ \tilde{A}_{1,t} \leq_{lr} \tilde{A}_{2,t} $, $ \tilde{X}_{1,t} \leq_{lr} \tilde{X}_{2,t} $ for all $ t\geq 0 $ \cite[Theorem 1.C.17]{SS:07}. It follows from  \cite[Theorem 1.C.20]{SS:07} that
\[
\frac{\mathbb{E}(\tilde{X}_{1,t}^{2})}{\mathbb{E}(\tilde{X}_{1,t})} \leq \frac{\mathbb{E}(\tilde{X}_{2,t}^{2})}{\mathbb{E}(\tilde{X}_{2,t})},
\]
for all $ t \geq 0 $. From equation (\ref{PS:Eq2}), the variance of the parasite burden of the second predator will be greater than that of the first, and, as the expected parasite burdens are equal, the parasite burden of the second predator will also have a greater index of dispersion.

\bibliographystyle{vancouver}
\bibliography{BL}

\begin{thebibliography}{10}

\bibitem{Crofton:71}
Crofton HD.
\newblock A quantitative approach to parasitism.
\newblock Parasitology. 1971;62:179--193.

\bibitem{WBDMPRRS:02}
Wilson K, Bj{\o}rnstad ON, Dobson AP, Merler S, Poglayen G, Randolph SE, et~al.
\newblock Heterogeneities in macroparasite infections: patterns and processes.
\newblock In: Hudson PJ, Rizzoli A, Grenfell BT, Heesterbeek H, Dobson AP,
  editors. The Ecology of Wildlife Diseases. Oxford University Press; 2002. p.
  6--44.

\bibitem{AGCH:82}
Anderson RM, Gordon DM, Crawley MJ, Hassell MP.
\newblock Variability in the abundance of animal and plant species.
\newblock Nature. 1982;296:245--248.

\bibitem{Galvani:03}
Galvani AP.
\newblock Immunity, anitigenic heterogeneity, and aggregation of Helminth
  parasites.
\newblock J Parasitol. 2003;89:232--241.

\bibitem{KA:79}
Keymer AE, Anderson RM.
\newblock The dynamics of infection of \emph{Tribolium confusum} by
  \emph{Hymenolepis diminuta}: the influence of infective-stage density and
  spatial distribution.
\newblock Parasitology. 1979;79:195--207.

\bibitem{JH:14}
Johnson PTJ, Hoverman JT.
\newblock Heterogeneous hosts: how variation in host size, behaviour and
  immunity affects parasite aggregation.
\newblock J Anim Ecol. 2014;83:1103--1112.

\bibitem{TDS:13}
Tadiri CP, Dargent F, Scott ME.
\newblock Relative host body condition and food availability influence epidemic
  dynamics: a \emph{Poecilia reticulata-Gyrodactylus turnbulli} host-parasite
  model.
\newblock Parasitology. 2013;140:343--351.

\bibitem{LGH:14}
Luong LT, Grear DA, Hudson PJ.
\newblock Manipulation of host-resource dynamics impacts transmission of
  trophic parasites.
\newblock Int J Parasitol. 2014;44:737--742.

\bibitem{KSKMSHM:06}
Krasnov BR, Stanko M, Khokhlova IS, Mosansky L, Shenbrot GI, Hawlena H, et~al.
\newblock Aggregation and species coexistence in fleas parasitic on small
  mammals.
\newblock Ecography. 2006;29:159--168.

\bibitem{JK:88}
Janovy J Jr, Kutish GW.
\newblock A model of encounters between host and parasite populations.
\newblock J Theor Biol. 1988;134:391--401.

\bibitem{SGD:98}
Shaw DJ, Grenfell BT, Dobson AP.
\newblock Patterns of macroparasite aggregation in wildlife host populations.
\newblock Parasitology. 1998;117:597--610.

\bibitem{Poulin:07}
Poulin R.
\newblock Evolutionary Ecology of Parasites.
\newblock 2nd ed. Princeton University Press; 2007.

\bibitem{DM:91}
Dobson AP, Merenlender A.
\newblock Coevolution of macroparasites and their hosts.
\newblock In: Toft CA, Sihilmann AE, editors. Parasitism: Coexistence or
  Conflict? Oxford University Press; 1991. p. 83--101.

\bibitem{SD:95}
Shaw DJ, Dobson AP.
\newblock Patterns of macroparasite abundance and aggregation in wildlife
  populations: a quantitative review.
\newblock Parasitology. 1995;111:S111--S133.

\bibitem{Lester:12}
Lester RJG.
\newblock Overdispersion in marine fish parasites.
\newblock J Parasitol. 2012;98:718--721.

\bibitem{CL:66}
Cox DR, Lewis PAW.
\newblock The Statistical Analysis of Series of Events.
\newblock Methuen; 1966.

\bibitem{Scott:87}
Scott ME.
\newblock Temporal changes in aggregation: a laboratory study.
\newblock Parasitology. 1987;94:583--595.

\bibitem{Poulin:93}
Poulin R.
\newblock The disparity between observed and uniform distributions: a new look
  at parasite aggregation.
\newblock Int J Parasitol. 1993;23:937--944.

\bibitem{Taylor:61}
Taylor LR.
\newblock Aggregation, variance and the mean.
\newblock Nature. 1961;189:732--735.

\bibitem{BM:55}
Bishop YM, Margolis L.
\newblock A statistical examination of \emph{Anisakis} Larvae (Nematoda) in
  Herring (\emph{Clupea pallasi}) of the British Columbia coast.
\newblock J Fish Res Board Can. 1955;12:571--592.

\bibitem{AKKK:97}
Amundsen PA, Kristoffersen R, Knudsen R, Klemetsen A.
\newblock Infection of Salmincola edwardsii (\emph{Copepoda: Lernaeopodidae})
  in an age-structured population of Arctic charr --- a long-term study.
\newblock J Fish Biol. 1997;51:1033--1046.

\bibitem{EEABSM:15}
Emre Y, Emre N, Aydogdu A, Bu{\v{s}}eli{\'{c}} I, Smales LR, Mladineo I.
\newblock Population dynamics of two diplectanid species (Monogenea)
  parasitising sparid hosts (Sparidae).
\newblock Parasitol Res. 2015;114:1079--1086.

\bibitem{Weinstein:72}
Weinstein M.
\newblock Studies on the relationship between \emph{Sagitta elegans} Verrill
  and its endoparasites in the Southwestern Gulf of St. Lawrence.
\newblock Marine Sciences Centre, McGill University; 1972.

\bibitem{Pearre:76}
Pearre SJ.
\newblock Gigantism and partial parasitic castration of chaetognatha infected
  with larval trematodes.
\newblock J Mar Biol Ass UK. 1976;56:503--513.

\bibitem{JK:85}
Jarling C, Kapp H.
\newblock Infestation of Atlantic chaetognaths with helminths and ciliates.
\newblock Dis Aquat Org. 1985;1:23--28.

\bibitem{RLC:99}
Roca V, Lafuente M, Carbonell E.
\newblock Helminth communities in Audouin's Gulls, \emph{Larus audouinii} from
  Chafarinas Islands (Western Mediterranean).
\newblock J Parasitol. 1999;85:984--986.

\bibitem{SSOM:03}
Simkova A, Sitko J, Okulewicz J, Morand S.
\newblock Occurrence of intermediate hosts and structure of digenean
  communities of the black-headed gull, \emph{Larus ridibundus} (L.).
\newblock Parasitology. 2003;126:69--78.

\bibitem{MAA:11}
Monteiro CM, Amato JFR, Amato SB.
\newblock Helminth parasitism in the Neotropical cormorant, \emph{Phalacrocorax
  brasilianus}, in southern Brazil: effect of host size, weight, sex, and
  maturity state.
\newblock Parasitol Res. 2011;109:849--855.

\bibitem{VMGRFL:11}
Violante-Gonz{\'a}lez J, Monks S, Gil-Guerrero S, Rojas-Herrera A, Flores-Garza
  R, Larumbe-Mor{\'a}n E.
\newblock Parasite communities of the neotropical cormorant \emph{Phalacrocorax
  brasilianus} (Gmelin) (Aves, Phalacrocoracidae) from two coastal lagoons in
  Guerrero state, Mexico.
\newblock Parasitol Res. 2011;109:1303--1309.

\bibitem{RSKKP:13}
Rzad I, Sitko J, Kavetska K, Kalisinska E, Panicz R.
\newblock Digenean communities in the tufted duck [\emph{Aythya fuligula} (L.,
  1758)] and greater scaup [\emph{A. marila} (L., 1761)] wintering in the
  north-west of Poland.
\newblock J Helminthol. 2013;87:230--239.

\bibitem{SPCC:15}
Sherrard-Smith E, Perkins SE, Chadwick EA, Cable J.
\newblock Spatial and seasonal factors are key determinants in the aggregation
  of helminths in their definitive hosts: \emph{Pseudamphistomum truncatum} in
  otters (\emph{Lutra lutra}).
\newblock Int J Parasitol. 2015;45:75--83.

\bibitem{RHH:95}
Rohde K, Hayward C, Heap M.
\newblock Aspects of the ecology of metazoan ectoparasites of marine fishes.
\newblock Int J Parasitol. 1995;25:945--970.

\bibitem{TBDWCP:11}
Tremblay-Boyer L, Gascuel D, Watson R, Christensen V, Pauly D.
\newblock Modelling the effects of fishing on the biomass of the world's oceans
  from 1950 to 2006.
\newblock Mar Ecol Prog Ser. 2011;442:169--185.

\bibitem{LBH:85}
Lester RJG, Barnes A, Habib G.
\newblock Parasites of skipjack tuna \emph{Katsuwonus pelamis}: fishery
  implications.
\newblock Fishery Bulletin. 1985;83:343--356.

\bibitem{WL:06}
Williams RE, Lester RJG.
\newblock Stock structure of Spanish mackerel \emph{Scomberomorus commerson}
  along the Australian east coast deduced from parasite data.
\newblock J Fish Biol. 2006;68:1707--1712.

\bibitem{FKNRJK:11}
Ferguson JA, Koketsu W, Ninomiya I, Rossignol PA, Jacobson KC, Kent ML.
\newblock Mortality of coho salmon (\emph{Oncorhynchus kisutch}) associated
  with burdens of multiple parasite species.
\newblock Int J Parasitol. 2011;41:1197--1205.

\bibitem{LM:94}
Lenski RE, May RM.
\newblock The evolution of virulence in parasites and pathogens: reconciliation
  between two competing hypotheses.
\newblock J Theor Biol. 1994;169:253--265.

\bibitem{Moller:76}
Moller H.
\newblock Reduction of the intestinal parasite fauna of marine fishes in
  captivity.
\newblock J Mar Biol Assoc UK. 1976;56:781--785.

\bibitem{SN:84}
Scott ME, Nokes DJ.
\newblock Temperature-dependent reproduction and survival of \emph{Gyrodactylus
  bullatarudis} (Monogenea) on guppies (\emph{Poecilia reticulata}).
\newblock Parasitology. 1984;89:221--227.

\bibitem{MB:69}
Margolis L, Boyce NP.
\newblock Life span, maturation, and growth of two hemiurid trematodes,
  \emph{Tubulovesicula lindbergi} and \emph{Lecithaster gibbosus}, in Pacific
  salmon (genus \emph{Oncorhynchus}).
\newblock J Fish Res Board Can. 1969;26:893--907.

\bibitem{AM:78}
Anderson RM, May RM.
\newblock Regulation and stability of host-parasite population interactions I.
  Regulatory processes.
\newblock J Anim Ecol. 1978;47:219--247.

\bibitem{Isham:95}
Isham V.
\newblock Stochastic models of host-macroparasite interaction.
\newblock Ann Appl Probab. 1995;5:720--740.

\bibitem{BP:00}
Barbour AD, Pugliese A.
\newblock On the variance-to-mean ratio in models of parasite distributions.
\newblock Adv Appl Probab. 2000;32:701--719.

\bibitem{ST:05}
Shoji J, Tanaka M.
\newblock Daily ration and prey size of juvenile piscivore Japanese Spanish
  mackerel.
\newblock J Fish Biol. 2005;67:1107--1118.

\bibitem{GKFM:09}
Griffiths SP, Kuhnert PM, Fry GF, Manson FJ.
\newblock Temporal and size-related variation in the diet, consumption rate,
  and daily ration of mackerel tuna (\emph{Euthynnus affinis}) in neritic
  waters of eastern Australia.
\newblock ICES J Mar Sci. 2009;66:720--733.

\bibitem{SDPF:14}
Stevens DA, Dunn MR, Pinkerton MH, Forman JS.
\newblock Diet of Antarctic toothfish (\emph{Dissostichus mawsoni}) from the
  continental slope and oceanic features of the Ross Sea region, Antarctica.
\newblock Antarct Sci. 2014;26:502--512.

\bibitem{LGL:15}
Logan JM, Golet WJ, Lutcavage ME.
\newblock Diet and condition of Atlantic bluefin tuna (\emph{Thunnus thynnus})
  in the Gulf of Maine, 2004–2008.
\newblock Environ Biol Fish. 2015;98:1411--1430.

\bibitem{SS:07}
Shaked M, Shanthikumar JG.
\newblock Stochastic Orders.
\newblock Springer; 2007.

\bibitem{Grandell:97}
Grandell J.
\newblock Mixed Poisson Processes.
\newblock Chapman \& Hall; 1997.

\bibitem{KP:09}
Kozubowski T, Podgorski K.
\newblock Distributional properties of the negative binomial L{\'e}vy process.
\newblock Probability and Mathematical Statistics. 2009;29:43--71.

\bibitem{Sale:02}
Sale PF, editor.
\newblock Coral Reef Fishes: Dynamics and Diversity in a Complex Ecosystem.
\newblock Academic press; 2002.

\bibitem{Roubal:95}
Roubal FR.
\newblock Changes in monogenean and copepod infestation on captive
  \emph{Acanthopagrus australis} (Sparidae).
\newblock J Fish Biol. 1995;46:423--431.

\bibitem{Poulin:13}
Poulin R.
\newblock Explaining variability in parasite aggregation levels among host
  samples.
\newblock Parasitology. 2013;140:541--546.

\bibitem{TW:82}
Taylor LR, Woiwod IP.
\newblock Comparative synoptic dynamics I. Relationships between inter- and
  intra-specific spatial and temporal variance/mean population parameters.
\newblock J Anim, Ecol. 1982;51:879--906.

\bibitem{AG:82}
Anderson RM, Gordon DM.
\newblock Processes influencing the distribution of parasite numbers within
  host populations with special emphasis on parasite-induced host mortalities.
\newblock Parasitology. 1982;85:373--398.

\bibitem{Feller:43}
Feller W.
\newblock On a general class of ``contagious'' distributions.
\newblock Ann Math Statist. 1943;14:389--400.

\end{thebibliography}

\newpage 

\begin{figure}
\includegraphics[width=7.5cm]{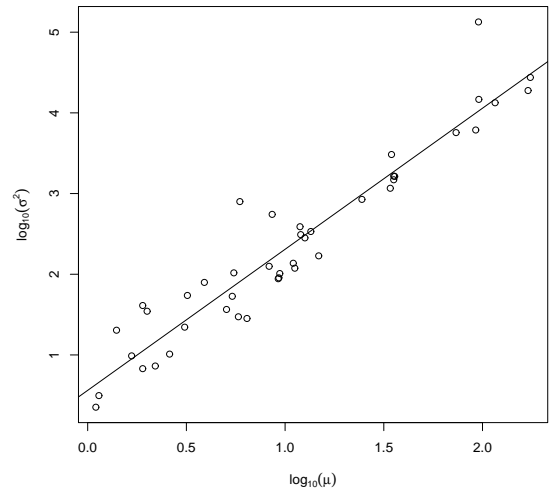}
\caption{Plot of the means and variances of \emph{Corynosoma} cystacanths in 43 fish samples from 15 publications. The line of best fit plotted is $ \log_{10}(\sigma^2) = 1.7484 \log_{10}(\mu) + 0.5609 $.} \label{Fig0}
\end{figure}

\begin{figure}
\includegraphics[width=7.5cm]{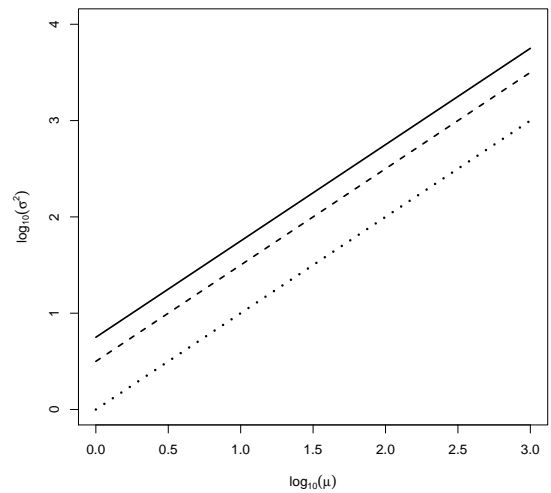}
\caption{Plot of log variance against log mean of the parasite burden. The parasite burden of prey (dotted line) follows a Poisson process. The first predator (dashed line) consumes the prey, and second predator (solid line) consumes the first predator.}\label{FigFPS}
\end{figure}

\begin{figure}
\includegraphics[width=7.5cm]{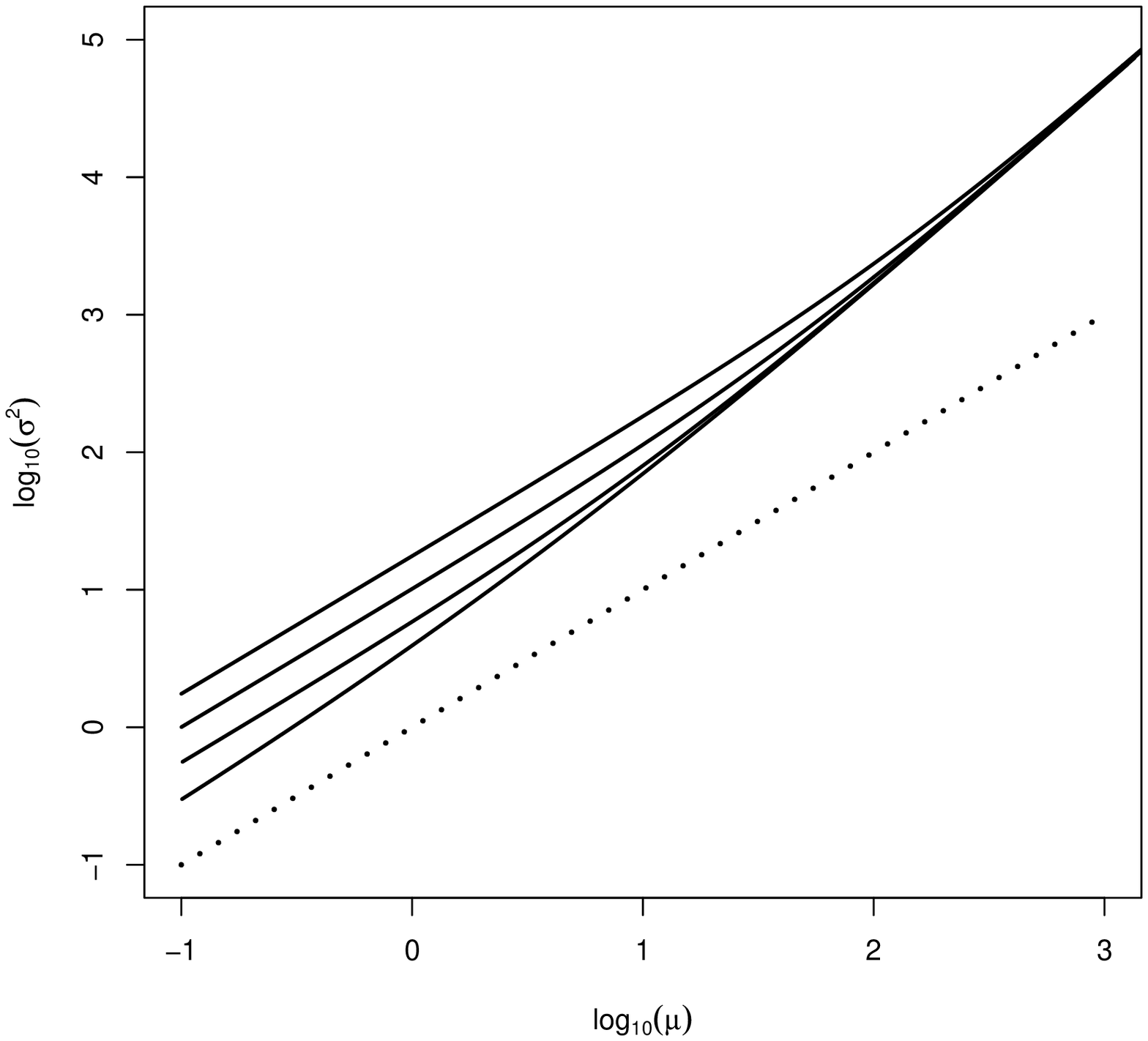}
\includegraphics[width=7.5cm]{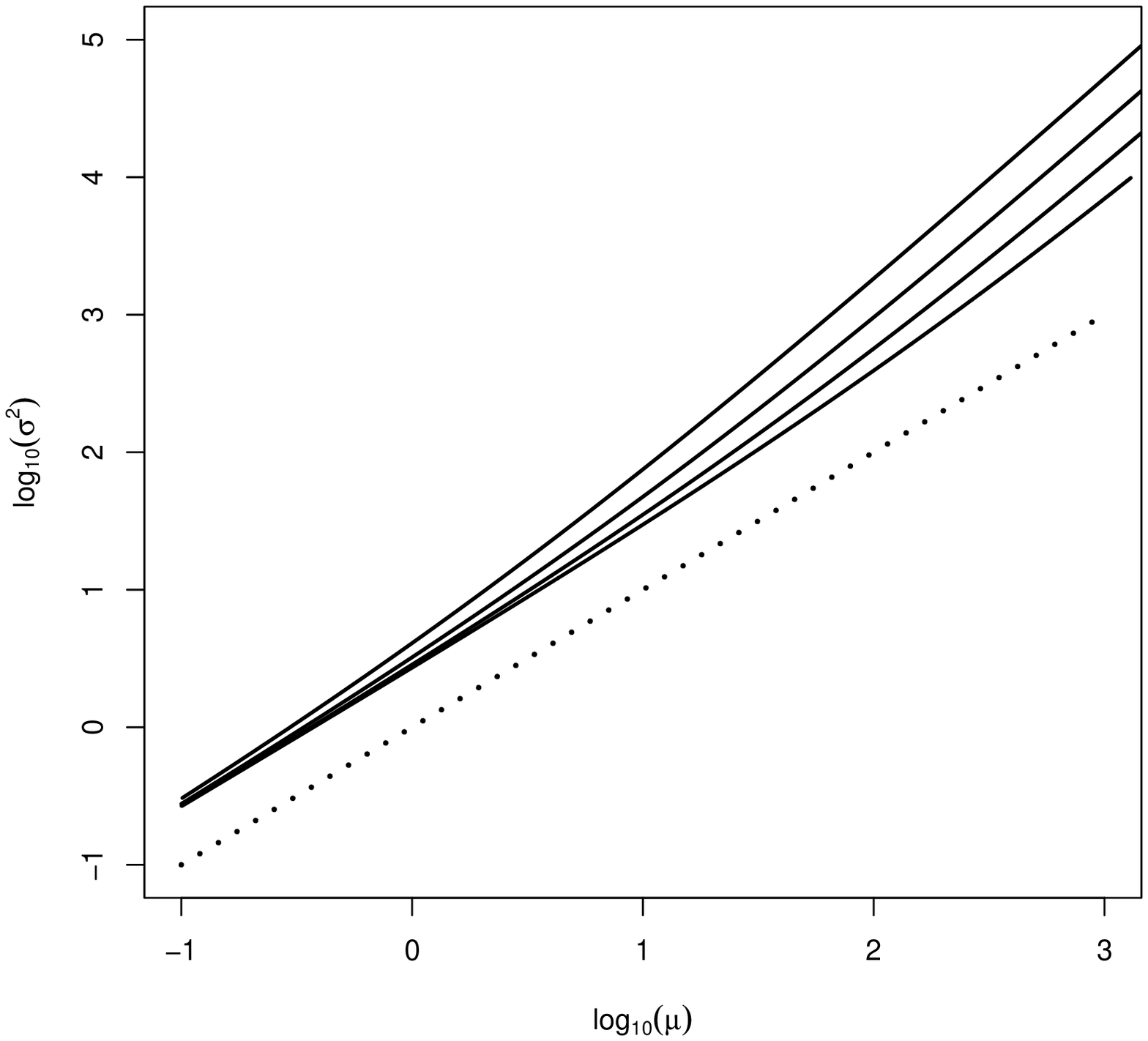}
\caption{Change in variance with changes in predator age: (Left) Plot of $ \log_{10}(\sigma^{2}) $ against $ \log_{10}(\mu) $ for $ \delta = 1$ and $ \alpha = (0.1,2,5,10) $.  The curves increase with $ \alpha $, but become less steep over this range of $ \mu $. (Right)  Plot of $ \log(\sigma^{2}) $ against $ \log(\mu) $ for $ \alpha = 0.1$ and $ \delta = (0.01,0.05,0.25,1.25) $. The curves increase with $ \delta $. For both plots, the dotted line corresponds to $ \log_{10}(\sigma^2) = \log_{10}(\mu) $.  }\label{Fig1}
\end{figure}

\begin{figure}
\includegraphics[width=7.5cm]{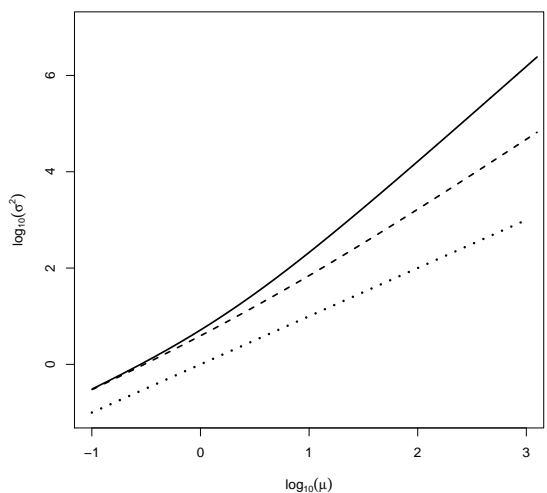}
\caption{Plot of $ \log_{10}(\sigma^{2}) $ against $ \log_{10}(\mu) $ with  $ \mathbb{E}(\lambda) =1,\ \mathbb{E}(\lambda^{2})=1.5,\ \beta+\gamma = 1$ and $ \alpha = 0.1 $ for both predators. For the first predator (dashed line) $ \delta_{1}(t) = t $ and $ \psi_{1}(t) = 1 $, and for the second predator (solid line) $ \delta_{2}(t) = t + t^{2} $ and $ \psi_{2}(t) = (0.1 + t)/(0.1 + t + t^{2}) $.  The dotted line corresponds to $ \log_{10}(\sigma^2) = \log_{10}(\mu) $. } \label{Fig2}
\end{figure}

\newpage

\begin{table}
\caption{Index of dispersion at a mean of 10 parasites per host ($ I_{10}$) for 28 parasite groups, from teleosts unless otherwise stated. Values for slope and intercept taken from log variance/log mean graph.}
\begin{tabular}{lrrlllr}
Parasite group	& No. of & No. of & slope & intercept  & $R^2 $ & $I_{10}$ \\
& papers & data & & &  & \\ 
& & points & & & & \\
\hline						
Adult cestodes in elasmobranch &	1 & 7 & 2.4 & 0.41  & 0.89 & 64.6 \\
Juvenile \emph{Contracaecum} spp. & 15 & 24 & 1.96 & 0.67  & 0.83 & 42.7 \\
Adult Digenea excl. hemiurids	& 39 & 122 & 2.07 & 0.56 & 0.87 & 42.7 \\
Adult Digenea from fish in birds \& otters & 6 & 29 & 1.43 & 1.17 & 0.73 & 39.8 \\
Diplectanids & 5 & 29 & 1.82 & 0.64 & 0.79  & 28.8 \\
Non-diplostome metacercariae & 25 & 86 & 1.71 & 0.75 & 0.88  & 28.8 \\
Diplostome metacercariae & 21 & 69 & 1.62 & 0.83 & 0.72  & 28.2\\
Adult anisakids & 6	& 12 & 1.78 & 0.65 & 0.82  & 26.9 \\
Adult spirurids	& 22 & 37 & 1.52 & 0.88 & 0.69 & 25.7 \\
Juvenile \emph{Hysterothylacium} & 14 & 28 & 2.02 & 0.39 & 0.93 & 25.7 \\
Adult \emph{Corynosoma} in seals & 1 & 5 & 2.03 & 0.35 & 0.94 & 24.0 \\
Adult Acanthocephala & 32 & 108 & 1.75 & 0.62 & 0.88 & 23.4 \\
\emph{Grillotia} blastocysts & 13 & 37 & 1.66 & 0.71 & 0.95  & 23.4\\
Ancyrocephalids & 5 & 13 & 1.75	& 0.62 & 0.92 & 23.4 \\
Adult hemiurids & 20 & 62 & 1.8 & 0.56 & 0.85 & 22.9 \\
Tetraphyllidean metacestodes & 9 & 17 & 1.43 & 0.91 & 0.75 & 21.9 \\
Adult cestodes	 & 24 & 64 & 1.71 & 0.63 & 0.84  & 21.9 \\
\emph{Corynosoma} cystacanths & 15 & 43 & 1.75 & 0.56 & 0.91  & 20.4 \\
Juvenile \emph{Anisakis simplex} & 6	& 24 & 1.74 & 0.5 & 0.91  & 17.4 \\
Diphyllobothriid plerocercoids	& 7 & 27 &  1.61 & 0.62 & 0.88 & 17.0 \\
Adult ascaridids & 11 & 17 & 1.92 & 0.25  & 0.81 & 14.8 \\
Non-caligoid copepods & 14 & 55 & 1.73 & 0.43 & 0.92 & 14.5 \\
Polyopisthocotyleans & 11 & 37 & 1.45 & 0.59 & 0.87 & 11.0 \\
Dactylogyrids	& 13 & 77 & 1.55 & 0.45 & 0.9  & 10.0\\
Caligoids	& 11 & 45 & 1.5 & 0.5 & 0.71  & 10.0 \\
Capsalids	& 4 & 17 & 1.55 & 0.39 & 0.83  & 8.7 \\
Cystacanths in inverts (extrapolated) & 1 & 6 & 1.15 & 0.28 & 0.92 & 2.7\\
Hemiurids in inverts (extrapolated) & 3 & 11 & 1.14 & 0.27 & 0.85 & 2.6 \\
\hline
Total & 354 & 1108
\end{tabular}		
\label{Tab1}		
\end{table}

\end{document}